\documentclass{article}%
\usepackage{amsmath}
\usepackage{amsfonts}
\usepackage{amssymb}
\usepackage{graphicx}%
\newtheorem{theorem}{Theorem}

\newtheorem{definition}[theorem]{Definition}

\newtheorem{proposition}[theorem]{Proposition}
\newtheorem{remark}[theorem]{Remark}

\newenvironment{proof}[1][Proof]{\noindent\textbf{#1.} }{\ \rule{0.5em}{0.5em}}
\newcommand{\bpartial}{\mathop{\partial\kern -4pt\raisebox{.8pt}{$|$}}}
\newcommand{\bra}{\mathopen{[\kern-1.6pt[}}
\newcommand{\ket}{\mathclose{]\kern-1.5pt]}}
\newcommand{\bbra}{\mathopen{[\kern-2.2pt[\kern-2.3pt[}}
\newcommand{\bket}{\mathclose{]\kern-2.1pt]\kern-2.3pt]}}
\newcommand{\sle}{\mbox{\bfseries\slshape e}}
\newcommand{\slg}{\mbox{\bfseries\slshape g}}
\newcommand{\sslg}{\mbox{\tiny \bfseries\slshape g}}

\newcommand{\slR}{\mbox{\bfseries\slshape R}}

\makeindex
\begin{document}

\title{Rigorous Formulation of Duality in Gravitational Theories}
\author{Rold\~{a}o da Rocha$^{(1)}$ and Waldyr A. Rodrigues Jr.$^{(2)}$\\$^{(1)}\hspace{-0.05cm}${\footnotesize Centro de Matem\'{a}tica,
Computa\c{c}\~{a}o e Cogni\c{c}\~{a}o}\\{\footnotesize Universidade Federal do ABC, 09210-170, Santo Andr\'e, SP,
Brazil}\\{\small \texttt{roldao.rocha@ufabc.edu.br}}\\$^{(2)}\,\hspace{-0.1cm}${\footnotesize Institute of Mathematics, Statistics
and Scientific Computation}\\{\footnotesize IMECC-UNICAMP CP 6065}\\{\footnotesize 13083-859 Campinas, SP, Brazil}\\{\small \texttt{walrod@ime.unicamp.br, walrod@mpc.com.br}}}
\maketitle

\begin{abstract}
In this paper we evince a rigorous formulation of duality in gravitational
theories where an Einstein like equation is valid, by providing the conditions
under which $\underset{%
\sslg
}{\star}\mathcal{T}^{\alpha}$ and $\underset{%
\sslg
}{\star}\mathcal{R}_{\beta}^{\alpha}$ may be considered as the torsion and
curvature $2$-forms associated with a connection $D^{\prime}$, part of a
Riemann-Cartan structure $(M,%
\slg
^{\prime},D^{\prime})$, in the cases $%
\slg
^{\prime}=%
\slg
$ and $%
\slg
^{\prime}\neq%
\slg
$, once $\mathcal{T}^{\alpha}$ and $\mathcal{R}_{\beta}^{\alpha}$ are the
torsion and curvature $2$-forms associated with a connection $D$ part of a
Riemann-Cartan structure $(M,%
\slg
,D)$. A new form for the Einstein equation involving the dual of the Riemann
tensor of $D$ is also provided, and the result is compared with others
appearing in the literature.

\end{abstract}


\section{Introduction}

There has been a number of papers trying to put in evidence a possible analogy
between electromagnetism and gravitation, in order to elicit a gravitational
analogue for the magnetic monopole that appears in the generalized Maxwell
equations with magnetic and electric currents\footnote{In such theory, see,
e.g., \cite{maia,rorod} which uses two potentials, the electric and magnetic
currents are phenomenological, i.e., the magnetic current is not a result of a
$U(1)$ gauge theory formulated in a nontrivial base spacetime. So, in the
theory which uses two potentials there are no Dirac strings at all.
Unfortunately, this result is sometimes overlooked in presentations of the
monopole theory and in the proposed gravitational analogies of that concept.}.
Among the old ones we quote\footnote{One of the motivations of \cite{dowro}
was eventually to obtain a quantization of mass.}
\cite{dowro,mignani,mignani1}. Ten years ago Nieto \cite{nieto} developed an
analogue of S-duality\footnote{Duality and S-duality have been also studied
extensively in non abelian gauge theories, see, e.g.
\cite{quevedo,maia1,witten} and references therein.} for linearized gravity in
($3+1$) dimensions (see, also \cite{garcia,nieto2,garcia2}) and
generalizations of that idea of \textit{duality} for gravitational theories in
more dimensions appear, e.g., in
\cite{argueh,argueh2,barnich,beroo,bber1,bunster,forder}. In particular for
the case of gravity in ($3+1$) dimensions a set of equations has been proposed
for Einstein equations, Bianchi identities, and their duals, although mainly
used in the linear approximation. The main aim of this note is to derive exact
equations that must be satisfied by the dual Einstein equations and for the
duals of the torsion and curvature $2$-forms of a general Riemann-Cartan
structure $(M,%
\slg
,D).$ We study in which conditions the dualized objects realize a
Riemann-Cartan structure $(M,%
\slg
,D^{\prime})$, or even a $(M,%
\slg
^{\prime},D^{\prime})$ one. In so doing we find that the correct field
equations for a dual theory (in a precise mathematical sense defined below),
are at variance with ones proposed in some of the above mentioned papers. In
so doing we hope that the present note be useful for those pursuing the
interesting ideas of duality in gravitational theories.

The paper, which uses an intrinsic formulation of the theories presented, is
organized as follows. In Section 2 we present some necessary preliminaries
that serve, besides the proposal of introducing our notation, also the one of
presenting what it is understood here by a Riemann-Cartan gravitational
theory. In this Section we review also the Bianchi identities for the torsion
and curvature $2$-forms $\mathcal{T}^{\alpha}$ and $\mathcal{R}_{\beta
}^{\alpha}$ of $(M,%
\slg
,D)$ in intrinsic and component forms, because those formulas for a
Riemann-Cartan theory are not well known as they deserve to be, and sometimes
concealed from the formalism. In Section 3 we introduce the Ricci  1-form
fields $\mathcal{R}^{\mu}$ and the Einstein 1-forms fields $\mathcal{G}^{\mu}$
\cite{rodcap2007}, and further prove a Proposition containing a formula that
relates the dual $\underset{%
\sslg
}{\star}\mathcal{R}^{\mu}$ of $\mathcal{R}^{\mu}$ to a sum, involving the
\textit{dual} of the Riemann tensor and an important formula for the dual
$\underset{%
\sslg
}{\star}\mathcal{G}^{\mu}$ of $\mathcal{G}^{\mu}$, that permits us to write
Einstein equations in a suggestive way concerning duality structures. In
Section 4 we provide the correct dual of Einstein equation in Riemann-Cartan
theory. In Section 5 we delve into the formalism under which conditions
$\underset{%
\sslg
}{\star}\mathcal{T}^{\alpha}$ and $\underset{%
\sslg
}{\star}\mathcal{R}_{\beta}^{\alpha}$ may be considered as the torsion and
curvature $2$-forms associated with a connection $D^{\prime}$ part of a
Riemann-Cartan structure $(M,%
\slg
,D^{\prime})$. Our result is then compared in Section 6 with the ones, e.g.,
in \cite{argueh}, which are then commented and analyzed in the present
context. In Section 7 we study the same problem as in Section 5 but this time
asking the conditions under which $\underset{%
\sslg
}{\star}\mathcal{T}^{\alpha}$ and $\underset{%
\sslg
}{\star}\mathcal{R}_{\beta}^{\alpha}$ may be considered as the torsion and
curvature $2$-forms associated to a connection $D^{\prime}$ part of a
Riemann-Cartan structure $(M,%
\slg
^{\prime},D^{\prime})$ with $%
\slg
^{\prime}\neq%
\slg
$. In Section 8 we present our conclusions. The paper contains some Appendices
reviewing the definition of the exterior covariant derivative of
\textit{indexed} form fields, the decomposition of the Riemann and Ricci
tensors of a general Riemann-Cartan structure $(M,%
\slg
,D)$, together with their respective similars for a Lorentzian structure
$(M,\mathring{%
\slg
},\mathring{D})$, needed to perceive some statements in the main text. There
is also an Appendix containing a collection of identities involving the
contraction of differential forms and Hodge duals used in the derivations hereon.

\section{Some Necessary Preliminaries}

We start this Section by recalling some germane facts concerning the
Riemann-Cartan structures and a particular and outstanding case of those
structures, the Lorentzian one, which serves for the purpose of fixing our
notations, besides other relevant properties and prominent applications. In
what follows a general Riemann-Cartan structure will be denoted by $(M,%
\slg
,D)$. Here $M$ is a $4$-dimensional Hausdorff, paracompact, connected, and
noncompact manifold, $%
\slg
\in\sec T_{2}^{0}M$ a metric tensor field of signature $(1,3)$, $D$ is a
connection on $M$. Also the connection $D$ is metric compatible, i.e., $ D%
\slg
=0$ and, moreover, for a general Riemann-Cartan structure the torsion and
curvature tensors\footnote{For the conventions used for those tensors in this
paper see the Appendix.} of $D$ --- denoted by $\mathcal{T}$ and $\mathbf{R}$
--- are non null. When $\mathcal{T}=0$ and $\mathbf{R} \neq0$, a
Riemann-Cartan structure is called a \textit{Lorentzian structure} and will be
denoted by $(M,%
\slg
,\mathring{D})$\footnote{The connection satisfying $\mathring{D}%
\slg
=0$ and $\mathbf{\mathcal{T}=}0$ is unique and is called the Levi-Civita
connection of $%
\slg
$.}. When $\mathbf{R}=0$ a Lorentzian structure is called Minkowski structure.
To present the definition of $\mathcal{T}$ and $\mathbf{R}$, and the
conventions used in this paper, first the torsion and curvature operations are introduced.

\begin{definition}
\label{torsion+curv} Let $\mathbf{u,v\in}\sec TM$. The \textit{torsion and
curvature operations} of a connection $D$, are respectively the mappings:
$\mathbf{\tau:}\sec TM \otimes\sec TM\rightarrow\sec TM$ and $\mathbf{\rho}:
\sec TM \otimes\sec TM\rightarrow\sec TM$ given by
\begin{align}
\mathbf{\tau}(\mathbf{u,v})  &  =\nabla_{\mathbf{u}}\mathbf{v}-\nabla
_{\mathbf{v}}\mathbf{u}-[\mathbf{u,v}],\label{top}\\
\mathbf{\rho(u,v)}  &  =\nabla_{\mathbf{u}}\nabla_{\mathbf{v}}-\nabla
_{\mathbf{v}}\nabla_{\mathbf{u}}-\nabla_{\lbrack\mathbf{u,v}]}. \label{cop}%
\end{align}

\end{definition}

\begin{definition}
Let $\mathbf{u,v,w}\in\sec TM$ and $\alpha\in\sec\Lambda^{1}T^{\ast}M$. The
torsion and curvature tensors of a connection $D$ are the mappings
$\mathcal{T}:\sec(\Lambda^{1}T^{\ast}M\otimes TM\otimes TM)\rightarrow
\mathbb{R}$ and $\mathbf{R}:\sec(TM\otimes\Lambda^{1}T^{\ast}M\otimes
TM\otimes TM)\rightarrow\mathbb{R}$ given by
\begin{align}
\mathcal{T}(\alpha,\mathbf{u,v}) &  =\alpha\left(  \mathbf{\tau}%
(\mathbf{u,v})\right)  ,\label{to op}\\
\mathbf{R}(\mathbf{w},\alpha,\mathbf{u,v}) &  =\alpha(\mathbf{\rho
(u,v)w}).\label{curv op}%
\end{align}

\end{definition}

Given an arbitrary moving frame $\{%
\sle
_{\alpha}\}$ on $TM$, let $\{\theta^{\rho}\}$ be the \textit{dual frame\/ of}
$\{%
\sle
_{\alpha}\}$ (i.e., $\theta^{\rho}(%
\sle
_{\alpha})=\delta_{\alpha}^{\rho}$). Let also $\{%
\sle
^{\alpha}\}$ be the reciprocal basis of $\{%
\sle
_{\beta}\}$, i.e., $%
\slg
(%
\sle
^{\alpha},%
\sle
_{\beta})=\delta_{\beta}^{\alpha}$ and let $\{\theta_{\alpha}\}$ be the
reciprocal basis of $\{\theta^{\rho}\}$, i.e., \ $\theta_{\alpha}(%
\sle
^{\beta})=\delta_{\alpha}^{\beta}$. We write
\begin{equation}%
\begin{array}
[c]{ccl}%
\lbrack%
\sle
_{\alpha}\mathbf{,}%
\sle
_{\beta}] & = & c_{\alpha\beta}^{\rho}%
\sle
_{\rho},\qquad\qquad\quad D_{%
\sle
_{\alpha}}%
\sle
_{\beta} = L_{\alpha\beta}^{\rho}%
\sle
_{\rho},
\end{array}
\end{equation}
where $c_{\alpha\beta}^{\rho}$ are the \textit{structure coefficients\/} of
the frame $\{%
\sle
_{\alpha}\}$ and $L_{\alpha\beta}^{\rho}$ are the \textit{connection
coefficients} in this frame. Then, the components of the torsion and curvature
tensors are given, respectively, by:
\begin{equation}%
\begin{array}
[c]{c}%
\mathcal{T}(\theta^{\alpha},%
\sle
_{\alpha}\mathbf{,}%
\sle
_{\beta})=T_{\alpha\beta}^{\rho}=L_{\alpha\beta}^{\rho}-L_{\beta\alpha}^{\rho
}-c_{\alpha\beta}^{\rho}\\
\mathbf{R(}%
\sle
_{\mu},\theta^{\alpha},%
\sle
_{\alpha}\mathbf{,}%
\sle
_{\beta}\mathbf{)}=R_{\mu}{}^{\rho}{}_{\!\alpha\beta}=%
\sle
_{\alpha}(L_{\beta\mu}^{\rho})-%
\sle
_{\beta}(L_{\alpha\mu}^{\rho})+L_{\alpha\sigma}^{\rho}L_{\beta\mu}^{\sigma
}-L_{\beta\sigma}^{\rho}L_{\alpha\mu}^{\sigma}-c_{\alpha\beta}^{\sigma
}L_{\sigma\mu}^{\rho}.
\end{array}
\label{585}%
\end{equation}
We can easily verify that defining
\begin{equation}
R_{\mu\nu\alpha\beta}:=g_{\mu\rho}R_{\mu}{}^{\rho}{}_{\!\alpha\beta}
\label{585a}%
\end{equation}
it follows that%
\begin{equation}
R_{\mu\nu\alpha\beta}=R_{\nu\mu\alpha\beta}=R_{\mu\nu\beta\alpha}.
\label{585b}%
\end{equation}

\begin{remark}
When the torsion tensor of $D$ is null, besides the symmetries given in
\emph{Eq.(\ref{585b})}, also the symmetry%
\begin{equation}
R_{\mu\nu\alpha\beta}=R_{\beta\alpha\mu\nu} \label{585c}%
\end{equation}
holds.
\end{remark}

Now, taking into account Eq.(\ref{585b}) we introduce also a \textquotedblleft
physically equivalent\textquotedblright\ Riemann tensor $%
\slR
$ by%
\begin{align}%
\slR
&  =\frac{1}{4}R_{\mu\nu\alpha\beta}\theta^{\mu}\wedge\theta^{\nu}%
\otimes\theta^{\alpha}\wedge\theta^{\beta}=\frac{1}{4}R^{\mu\nu}{}%
_{\!\alpha\beta}\theta_{\mu}\wedge\theta_{\nu}\otimes\theta^{\alpha}%
\wedge\theta^{\beta}\nonumber\\
&  =\frac{1}{4}R_{\mu\nu}^{\,\;\;\alpha\beta}\theta^{\mu}\wedge\theta^{\nu
}\otimes\theta_{\alpha}\wedge\theta_{\beta}.\label{586}%
\end{align}
In addition,
\begin{equation}%
\begin{array}
[c]{l}%
d\theta^{\rho}=-\frac{1}{2}c_{\alpha\beta}^{\rho}\theta^{\alpha}\wedge
\theta^{\beta},\qquad\qquad\quad D_{e_{\alpha}}\theta^{\rho}=-L_{\alpha\beta
}^{\rho}\theta^{\beta}%
\end{array}
\label{608}%
\end{equation}
where $\omega_{\beta}^{\rho}\in\sec\Lambda^{1}T^{\ast}M$ are the
\textit{connection 1-forms, } $\mathcal{T}^{\rho}\in\sec\Lambda^{2}T^{\ast}M$
are the \textit{torsion 2-forms} and $\mathcal{R}_{\beta}^{\rho}\in\sec
\Lambda^{2}T^{\ast}M$\textbf{\ }are the \textit{curvature 2-forms}, given
respectively by
\begin{equation}
\omega_{\beta}^{\rho}=L_{\alpha\beta}^{\rho}\theta^{\alpha},\quad
\quad\mathcal{T}^{\rho}=\frac{1}{2}T_{\alpha\beta}^{\rho}\theta^{\alpha}%
\wedge\theta^{\beta},\quad\quad\mathcal{R}_{\mu}^{\rho}=\frac{1}{2}R_{\mu}%
{}^{\rho}{}_{\!\alpha\beta}\theta^{\alpha}\wedge\theta^{\beta}.\label{620}%
\end{equation}
Multiplying Eqs.(\ref{585}) by $\frac{1}{2}\theta^{\alpha}\wedge\theta^{\beta
}$ and using Eqs.(\ref{608}) and~(\ref{620}), the \textit{Cartan's structure
equations} are derived:%
\begin{equation}%
\begin{array}
[c]{l}%
\mathcal{T}^{\rho}=d\theta^{\rho}+\omega_{\beta}^{\rho}\wedge\theta^{\beta
},\qquad\qquad\quad\mathcal{R}_{\mu}^{\rho}=d\omega_{\mu}^{\rho}+\omega
_{\beta}^{\rho}\wedge\omega_{\mu}^{\beta}.
\end{array}
\label{559}%
\end{equation}

\begin{definition}
A Riemann-Cartan spacetime is a pentuple $(M,%
\slg
,D,\tau_{%
\sslg
},\uparrow)$ where $(M,%
\slg
,D)$ is a Riemann-Cartan structure, and we suppose the existence of a global
$\tau_{%
\sslg
}\in\sec\Lambda^{4}T^{\ast}M$ \emph{(}which as well known defines an
orientation for $M$\emph{)}. Moreover, $\uparrow$ denotes that the
Riemann-Cartan structure is time oriented. See, e.g.,
\emph{\cite{rodcap2007,sawu}} for details.
\end{definition}

Pentuples $(M,%
\slg
,D,\tau_{%
\sslg
},\uparrow)$ represent gravitational fields in the so called Riemann-Cartan
theories. In the theory presented, e.g., in \cite{hehl1}, the equations of
motion are the Einstein equation,%
\[
\mathbf{G}=\mathbf{T,}%
\]
where $\mathbf{G}\in\sec T_{2}^{0}M$ is the Einstein tensor, $\mathbf{T}%
\in\sec T_{2}^{0}M$ is the \textit{canonical energy-momentum tensor} of the
matter fields, and the algebraic identity%
\begin{equation}
\mathbf{\Upsilon}_{\alpha\beta}=\mathbf{J}_{\alpha\beta}\mathbf{,} \label{e2}%
\end{equation}
where the $\mathbf{\Upsilon}_{\alpha\beta}\in\sec\Lambda^{1}T^{\ast}M$ are
such that their components are the so called modified torsion \textit{tensor
components}, and the $\underset{%
\sslg
}{\star}\mathbf{J}_{\alpha\beta}$ $\in\sec\Lambda^{3}T^{\ast}M$ are the spin
angular momentum densities of the matter fields\footnote{The components of
$\mathbf{J}_{\alpha\beta}\in\sec\Lambda^{1}T^{\ast}M$ are the standard (field
theory) canonical spin angular momentum of the matter fields. In the
Riemann-Cartan theory of \cite{hehl1}, since Eq.(\ref{e2}) is an algebraic
identity, it is possible to eliminate completely the torsion tensor from the
theory and to write an Einstein equation involving the Einstein tensor of the
Levi-Civita connection of $%
\slg
$ (using the decomposition presented in Appendix B) and a metric
energy-momentum tensor that is equivalent to the Belinfante symmetrization of
the canonical energy-momentum tensor of the theory.}. Also, the symbol
$\underset{%
\sslg
}{\star}$ denotes the Hodge star operator associated to the metric $%
\slg
$.

\begin{remark}
It is crucial to observe that for a general Riemann-Cartan structure
$\mathbf{G}=G_{\mu\nu}\theta^{\mu}\otimes\theta^{\nu}$ and $\mathbf{T}%
=T_{\mu\nu}\theta^{\mu}\otimes\theta^{\nu}$ are not symmetric, i.e.,
$G_{\mu\nu}\neq G_{\nu\mu}$ and $T_{\mu\nu}\neq T_{\nu\mu}$. We recall that
\begin{equation}
G_{\mu\nu}=R_{\mu\nu}-\frac{1}{2}g_{\mu\nu}R, \label{e3}%
\end{equation}
where $R_{\mu\nu}$ are the components of the Ricci tensor \emph{(}which, as
$G_{\mu\nu}$ are not symmetric\emph{)}
\begin{equation}
Ricci=R_{\mu\nu}\theta^{\mu}\otimes\theta^{\nu}:=R_{\mu}{}^{\rho}{}_{\!\rho
\nu}\theta^{\mu}\otimes\theta^{\nu}, \label{e4}%
\end{equation}
and $R=g^{\mu\nu}R_{\mu}{}^{\rho}{}_{\!\rho\nu}$ is the curvature scalar.
\end{remark}

It is also well known that in \textit{GRT} a gravitational field generated by
a given matter distribution (represented by a given energy-momentum tensor
$\mathbf{\mathring{T}}\in\sec T_{2}^{0}M$) is represented by a pentuple
$(M,\mathring{%
\slg
},\mathring{D},\tau_{\mathring{%
\sslg
}},\uparrow)$\footnote{In fact a gravitational field is defined by an
equivalence class of pentuples, where $(M,%
\slg
,D,\tau_{%
\sslg
},\uparrow)$ and $(M^{\prime},%
\slg
^{\prime},D^{\prime},\tau_{%
\sslg
}^{\prime},\uparrow^{\prime})$ are said to be equivalent if there is a
diffeomorphism $\mathtt{h}:M\rightarrow M^{\prime}$, such that $%
\slg
^{\prime}=\mathtt{h}^{\ast}%
\slg
$, $D^{\prime}=\mathtt{h}^{\ast}D$, $\tau_{%
\sslg
}^{\prime}=\mathtt{h}^{\ast}\tau_{%
\sslg
},\uparrow^{\prime}=\mathtt{h}^{\ast}\uparrow,$ (where $\mathtt{h}^{\ast}$
here denotes the pullback mapping). For more details, see, e.g., \cite{sawu,
rodcap2007}. With the above definition we exclude from our considerations
models with closed timelike curves, which according to our view are pure
science fiction.} and the equation of motion (Einstein equation) is given by%
\begin{equation}
\mathbf{\mathring{G}=\mathring{T}} \label{ee1}%
\end{equation}
and in this case the tensors $\mathbf{\mathring{G}}$ and $\mathbf{\mathring
{T}}$ are symmetric.

\begin{remark}
In the Appendix we review how to write the Riemann curvature tensor
(respectively the Einstein tensor) of a Riemann-Cartan structure $(M,%
\slg
,D)$ in terms of the Riemann curvature tensor (respectively the Einstein
tensor) of a Lorentzian structure $(M,\mathring{%
\slg
},\mathring{D})$. Those results are important for a proper understanding of
this paper.
\end{remark}

\subsection{The Bianchi Identities}

Given a general Riemann-Cartan structure $(M,%
\slg
,D)$ we have the following identities%
\begin{align}
\mathbf{D}\mathcal{T}^{\alpha}  &  =\mathcal{R}_{\beta}^{\alpha}\wedge
\theta^{\beta},\label{bianchi1}\\
\mathbf{D}\mathcal{R}_{\beta}^{\alpha}  &  =0, \label{bianchi2}%
\end{align}
known respectively as the first and second Bianchi identities (see e.g.,
\cite{choquet,rodcap2007}). In the above equations, $\mathbf{D}$ is the
exterior covariant derivative of indexed form fields \cite{been,rodcap2007},
whose precise definition is recalled in Appendix A. Now, the coordinate
expressions of Eqs.(\ref{bianchi1}) and (\ref{bianchi2}) can easily be found
and are respectively \cite{choquet,rodALB} written as%
\begin{align}%
{\displaystyle\sum\limits_{(\mu\alpha\beta)}}
R_{\mu}{}^{\rho}{}_{\alpha\beta}  &  =%
{\displaystyle\sum\limits_{(\mu\alpha\beta)}}
\left( D_{\mu}T_{\alpha\beta}^{\rho}-T_{\mu\beta}^{\kappa}T_{\kappa\alpha
}^{\rho}\right)  ,\label{bianchi11}\\%
{\displaystyle\sum\limits_{(\mu\nu\rho)}}
D_{\mu}R_{\beta\;\nu\rho}^{\;\alpha}  &  =%
{\displaystyle\sum\limits_{(\mu\nu\rho)}}
T_{\nu\mu}^{\kappa}R_{\beta\;\kappa\rho}^{\;\alpha}, \label{bianchi21}%
\end{align}
where $%
{\displaystyle\sum\limits_{(\mu\nu\rho)}}
$ denotes (as usual) the cyclic sum. For future use we observe that
\begin{equation}
\mathcal{R}_{\beta}^{\alpha}\wedge\theta^{\beta}=\frac{1}{3!}(R_{\mu}%
{}^{\alpha}{}_{\alpha\beta}+R_{\alpha}{}^{\alpha}{}_{\beta\mu}+R_{\beta}%
{}^{\alpha}{}_{\mu\alpha})\theta^{\mu}\wedge\theta^{\alpha}\wedge\theta
^{\beta}. \label{bia3}%
\end{equation}

\begin{remark}
For a Lorentz structure $(M,\mathring{%
\slg
},\mathring{D})$ the Bianchi identities reduce to
\end{remark}%
\begin{align}
\mathcal{\mathring{R}}_{\beta}^{\alpha}\wedge\theta^{\beta}  &  =0,\qquad
\qquad\quad\mathbf{D}\mathcal{\mathring{R}}_{\beta}^{\alpha} =0, \nonumber
\end{align}
or in components:%
\begin{align}%
{\displaystyle\sum\limits_{(\mu\alpha\beta)}}
R_{\mu}{}^{\rho}{}_{\alpha\beta}  &  =0, \qquad\qquad%
{\displaystyle\sum\limits_{(\mu\nu\rho)}}
D_{\mu}R_{\beta\;\nu\rho}^{\;\alpha} =0. \nonumber
\end{align}

\section{Ricci and Einstein 1-form fields}

Given $R_{\mu\nu}$ and $G_{\mu\nu}$, respectively the components of the Ricci
and Einstein tensors (in the general basis introduced above) we define the
Ricci ($\mathcal{R}^{\mu}\in\sec\Lambda^{1}T^{\ast}M$) and the Einstein
($\mathcal{G}^{\mu}\in\sec\Lambda^{1}T^{\ast}M$ ) $1$-form fields by%
\begin{align}
\mathcal{R}^{\mu}  & :=R_{\;\nu}^{\mu}\theta^{\nu},\qquad\qquad\qquad
\mathcal{G}^{\mu} := G_{\;\nu}^{\mu}\theta^{\nu}. \label{ei}%
\end{align}
For future use we introduce also the energy-momentum $1$-form fields
$\mathbf{T}^{\mu}\in\sec\Lambda^{1}T^{\ast}M$ by%
\begin{equation}
\mathbf{T}^{\mu}:=T_{\;\nu}^{\mu}\theta^{\nu}. \label{ti}%
\end{equation}
Also%
\begin{align}
\underset{%
\sslg
}{\star}\mathbf{T}^{\mu}  &  =T_{\;\nu}^{\mu}\underset{%
\sslg
}{\star}\theta^{\nu}\nonumber\\
&  =\frac{1}{3!}(T_{\;\nu}^{\mu}\sqrt{\left\vert \det%
\slg
\right\vert }g^{\nu\kappa}\epsilon_{\kappa\iota\lambda\sigma})\theta^{\iota
}\wedge\theta^{\kappa}\wedge\theta^{\sigma}. \label{tii}%
\end{align}

\begin{proposition}
The dual of the Ricci and Einstein $1$-form fields, i.e., $\underset{%
\sslg
}{\star}\mathcal{R}^{\alpha}\in\sec\Lambda^{3}T^{\ast}M$ and $\underset{%
\sslg
}{\star}\mathcal{G}^{\alpha}\in\sec\Lambda^{3}T^{\ast}M$ can be written as:
\begin{align}
\underset{%
\sslg
}{\star}\mathcal{R}^{\alpha}  &  =-\underset{%
\sslg
}{\star}\mathcal{R}_{\beta}^{\alpha}\wedge\theta^{\beta}=-\theta^{\beta}%
\wedge\underset{%
\sslg
}{\star}\mathcal{R}_{\beta}^{\alpha},\label{rii}\\
\underset{%
\sslg
}{\star}\mathcal{G}^{\rho}  &  =-\frac{1}{2}\mathcal{R}_{\alpha\beta}%
\wedge\underset{%
\sslg
}{\star}(\theta^{\alpha}\wedge\theta^{\beta}\wedge\theta^{\rho}). \label{eii}%
\end{align}
where $\mathcal{R}_{\mu}^{\rho}=\frac{1}{2}R_{\mu}{}^{\rho}{}_{\!\alpha\beta
}\theta^{\alpha}\wedge\theta^{\beta}$ and $\mathcal{R}_{\mu\rho}:=\frac{1}%
{2}R_{\mu\rho\alpha\beta}\theta^{\alpha}\wedge\theta^{\beta}$.
\end{proposition}

\begin{proof}
Using some of the identities in Appendix E we can write immediately%
\begin{align*}
\theta^{\rho}\wedge\underset{%
\sslg
}{\star}\mathcal{R}_{\mu\rho}  &  =-\underset{%
\sslg
}{\star}(\theta^{\rho}\underset{%
\sslg
}{\lrcorner}\mathcal{R}_{\mu\rho})\\
&  =-\underset{%
\sslg
}{\star}\frac{1}{2}[R_{\mu\rho\alpha\beta}\theta^{\rho}\underset{%
\sslg
}{\lrcorner}(\theta^{\alpha}\wedge\theta^{\beta})] =-\underset{%
\sslg
}{\star}(R_{\mu\rho\alpha\beta}g^{\rho\alpha}\theta^{\beta})\\
&  =-\underset{%
\sslg
}{\star}(R_{\mu\;\alpha\beta}^{\;\alpha}\theta^{\beta})=-\star(R_{\mu\beta
}\theta^{\beta})\\
&  =-\underset{%
\sslg
}{\star}\mathcal{R}_{\mu},
\end{align*}
and Eq.(\ref{rii}) is proved.

Now Eq.(\ref{eii}) is evinced. By taking some of the identities in Appendix E,
we can immediately write:%
\begin{align*}
\frac{1}{2}\mathcal{R}_{\alpha\beta}\wedge\underset{%
\sslg
}{\star}(\theta^{\alpha}\wedge\theta^{\beta}\wedge\theta^{\rho})  &
=-\frac{1}{2}\underset{%
\sslg
}{\star}[\mathcal{R}_{\alpha\beta}\underset{%
\sslg
}{\lrcorner}(\theta^{\alpha}\wedge\theta^{\beta}\wedge\theta^{\rho})]\\
&  =-\frac{1}{4}R_{\alpha\beta\iota\kappa}\underset{%
\sslg
}{\star}[(\theta^{\iota}\wedge\theta^{\kappa})\underset{%
\sslg
}{\lrcorner}(\theta^{\alpha}\wedge\theta^{\beta}\wedge\theta^{\rho})]
=-\frac{1}{4}R_{\alpha\beta\iota\kappa}\underset{%
\sslg
}{\star}[(\theta^{\iota}\underset{%
\sslg
}{\lrcorner}(\theta^{\kappa}\underset{%
\sslg
}{\lrcorner}(\theta^{\alpha}\wedge\theta^{\beta}\wedge\theta^{\rho}))]\\
&  =-\underset{%
\sslg
}{\star}(\mathcal{R}^{\rho}-\frac{1}{2}R\theta^{\rho}),
\end{align*}
and Eq.(\ref{eii}) is proved.
\end{proof}

\begin{remark}
Recall that
\begin{align}
\underset{%
\sslg
}{\star}\mathcal{R}_{\mu\rho}  &  :=\frac{1}{2}R_{\mu\rho\alpha\beta}%
\underset{%
\sslg
}{\star}(\theta^{\alpha}\wedge\theta^{\beta})\nonumber\\
&  =\frac{1}{2}R_{\mu\rho\alpha\beta}\frac{1}{2}\sqrt{\left\vert \det%
\slg
\right\vert }g^{\alpha\kappa}g^{\beta\iota}\epsilon_{\kappa\iota\lambda\sigma
}\theta^{\lambda}\wedge\theta^{\sigma} =\frac{1}{2}(\frac{1}{2}\sqrt
{\left\vert \det%
\slg
\right\vert }\epsilon_{\kappa\iota\lambda\sigma}R_{\mu\rho}^{\;\;\kappa\iota
})\theta^{\lambda}\wedge\theta^{\sigma}\nonumber\\
&  =\frac{1}{2}\mathsf{R}_{\mu\rho\lambda\sigma}^{\star}\theta^{\lambda}%
\wedge\theta^{\sigma}\label{p1}\\
&  =\frac{1}{2}R_{\mu\rho\lambda\sigma}^{\star}\sqrt{\left\vert \det%
\slg
\right\vert }\theta^{\lambda}\wedge\theta^{\sigma},\nonumber
\end{align}
with%
\begin{align}
\mathsf{R}_{\mu\rho\lambda\sigma}^{\star}  &  :=\frac{1}{2}\sqrt{\left\vert
\det%
\slg
\right\vert }\epsilon_{\kappa\iota\lambda\sigma}R_{\;\;\alpha\beta}%
^{\kappa\iota},\quad\quad\text{and}\qquad R_{\mu\rho\lambda\sigma}^{\star}
:=\frac{1}{2}\epsilon_{\kappa\iota\lambda\sigma}R_{\;\;\alpha\beta}%
^{\kappa\iota}, \label{p3}%
\end{align}
and so it follows that%
\begin{align}
\star\mathcal{R}_{\mu\rho}\wedge\theta^{\rho}  &  =\frac{1}{2}\mathsf{R}%
_{\mu\rho\lambda\sigma}^{\star}\theta^{\lambda}\wedge\theta^{\sigma}%
\wedge\theta^{\rho} =\frac{1}{2}\mathsf{R}_{\mu\rho\lambda\sigma}^{\star
}\theta^{\rho}\wedge\theta^{\lambda}\wedge\theta^{\sigma}\nonumber\\
&  =\frac{1}{2}\left(  \frac{1}{3}\mathsf{R}_{\mu\rho\lambda\sigma}^{\star
}\theta^{\rho}\wedge\theta^{\lambda}\wedge\theta^{\sigma}+\frac{1}%
{3}\mathsf{R}_{\mu\lambda\sigma\rho}^{\star}\theta^{\lambda}\wedge
\theta^{\sigma}\wedge\theta^{\rho}+\frac{1}{3}\mathsf{R}_{\mu\sigma\rho
\lambda}^{\star}\theta^{\sigma}\wedge\theta^{\rho}\wedge\theta^{\lambda
}\right) \label{p4}\\
&  =\frac{1}{3!}\left(  \mathsf{R}_{\mu\rho\lambda\sigma}^{\star}%
+\mathsf{R}_{\mu\lambda\sigma\rho}^{\star}+\mathsf{R}_{\mu\sigma\rho\lambda
}^{\star}\right)  \theta^{\rho}\wedge\theta^{\lambda}\wedge\theta^{\sigma
}\nonumber\\
&  =\frac{1}{3!}\left(  R_{\mu\rho\lambda\sigma}^{\star}+R_{\mu\lambda
\sigma\rho}^{\star}+R_{\mu\sigma\rho\lambda}^{\star}\right)  \sqrt{\left\vert
\det%
\slg
\right\vert }\theta^{\rho}\wedge\theta^{\lambda}\wedge\theta^{\sigma}.
\end{align}
and taking into account \emph{Eq.(\ref{rii})} it reads:%
\begin{equation}
\underset{%
\sslg
}{\star}\mathcal{R}_{\mu}=-\frac{1}{3!}\left(  R_{\mu\rho\lambda\sigma}%
^{\star}+R_{\mu\lambda\sigma\rho}^{\star}+R_{\mu\sigma\rho\lambda}^{\star
}\right)  \sqrt{\left\vert \det%
\slg
\right\vert }\theta^{\rho}\wedge\theta^{\lambda}\wedge\theta^{\sigma}.
\label{p5}%
\end{equation}

\end{remark}

\section{The Dual of Einstein Equation in Riemann-Cartan Theory}

We now return to Eq.(\ref{ei}) which in components can read\footnote{Take
notice that in Eq.(\ref{p6}) $R_{\mu\nu}$ and $T_{\mu\nu}$ are not symmetric.}%
\begin{equation}
R_{\mu\nu}-\frac{1}{2}g_{\mu\nu}R=T_{\mu\nu}\label{p6}%
\end{equation}
Multiplying this equation on both sides by $\theta^{\nu}$ and recalling the
definitions of the Ricci, Einstein, and the energy-momentum $1$-form fields
given above we have%
\begin{equation}
\mathcal{G}_{\mu}=\mathbf{T}_{\mu}.\label{p7}%
\end{equation}
Taking the dual of this equation we get%
\begin{equation}
\underset{%
\sslg
}{\star}\mathcal{G}_{\mu}=\underset{%
\sslg
}{\star}\mathcal{R}_{\mu}-\frac{1}{2}R\underset{%
\sslg
}{\star}\theta_{\mu}=\underset{%
\sslg
}{\star}\mathbf{T}_{\mu}\label{p8}%
\end{equation}
Taking Eq.(\ref{p5}) and Eq.(\ref{tii}) into account, then Eq.(\ref{p8}) can
be expressed as%
\begin{align*}
&  -\frac{1}{3!}\left(  R_{\mu\rho\lambda\sigma}^{\star}+R_{\mu\lambda
\sigma\rho}^{\star}+R_{\mu\sigma\rho\lambda}^{\star}+\frac{1}{2}R\delta_{\mu
}^{\kappa}\epsilon_{\kappa\rho\lambda\sigma}\right)  \sqrt{\left\vert \det%
\slg
\right\vert }\theta^{\rho}\wedge\theta^{\lambda}\wedge\theta^{\sigma}\\
&  =\frac{1}{3!}(T_{\mu\nu}g^{\nu\kappa}\epsilon_{\kappa\rho\lambda\sigma
})\sqrt{\left\vert \det%
\slg
\right\vert }\theta^{\rho}\wedge\theta^{\lambda}\wedge\theta^{\sigma},
\end{align*}
or equivalently%
\begin{equation}
\left(  R_{\mu\rho\lambda\sigma}^{\star}+R_{\mu\lambda\sigma\rho}^{\star
}+R_{\mu\sigma\rho\lambda}^{\star}+\frac{1}{2}R\epsilon_{\mu\rho\lambda\sigma
}\right)  =\epsilon_{\rho\lambda\sigma\kappa}T_{\mu}^{\;\kappa}.\label{p9}%
\end{equation}

\subsection{The Field and Structure Equations}

We now summarize the field and Bianchi identities for a Riemann-Cartan theory
where an Einstein-like equation holds. Those equations can be written
\textit{conveniently} in intrinsic and component forms respectively as:%
\begin{align}
\underset{%
\sslg
}{\star}\mathcal{G}_{\mu}  &  =\underset{%
\sslg
}{\star}\mathcal{R}_{\mu}-\frac{1}{2}R\underset{%
\sslg
}{\star}\theta_{\mu}=\underset{%
\sslg
}{\star}\mathbf{T}_{\mu}\label{f010}\\
\mathbf{D}\mathcal{T}^{\alpha}  &  =\mathcal{R}_{\beta}^{\alpha}\wedge
\theta^{\beta},\label{f01}\\
\mathbf{D}\mathcal{R}_{\beta}^{\alpha}  &  =0, \label{f011}%
\end{align}
\begin{align}
\left(  R_{\mu\rho\lambda\sigma}^{\star}+R_{\mu\lambda\sigma\rho}^{\star
}+R_{\mu\sigma\rho\lambda}^{\star}+\frac{1}{2}R\epsilon_{\mu\rho\lambda\sigma
}\right)   &  =\epsilon_{\rho\lambda\sigma\kappa}T_{\mu}^{\;\kappa
}\Longleftrightarrow G_{\mu\nu}=T_{\mu\nu},\label{f10}\\%
{\displaystyle\sum\limits_{(\mu\alpha\beta)}}
R_{\mu\rho\alpha\beta}  &  =%
{\displaystyle\sum\limits_{(\mu\alpha\beta)}}
\left(D_{\mu}T_{\rho\alpha\beta}-T_{\mu\beta}^{\kappa}T_{\rho\kappa\alpha
}\right)  ,\label{f1}\\%
{\displaystyle\sum\limits_{(\mu\nu\rho)}}
D_{\mu}R_{\beta\alpha\nu\rho}  &  =%
{\displaystyle\sum\limits_{(\mu\nu\rho)}}
T_{\nu\mu}^{\kappa}R_{\beta\alpha\kappa\rho}. \label{f11}%
\end{align}
In a \textit{GRT }model it follows that
\begin{align}
\left(\mathring{R}_{\mu\rho\lambda\sigma}^{\star}+\mathring{R}_{\mu
\lambda\sigma\rho}^{\star}+\mathring{R}_{\mu\sigma\rho\lambda}^{\star}%
+\frac{1}{2}\mathring{R}\epsilon_{\mu\rho\lambda\sigma}\right)   &
=\epsilon_{\rho\lambda\sigma\kappa}\mathring{T}_{\mu}^{\;\kappa}%
\Leftrightarrow G_{\mu\nu}=T_{\mu\nu},\nonumber\\%
{\displaystyle\sum\limits_{(\mu\alpha\beta)}}
\mathring{R}_{\mu\rho\alpha\beta}  &  =0, \qquad\qquad%
{\displaystyle\sum\limits_{(\mu\nu\rho)}}
\mathring{D}_{\mu}\mathring{R}_{\beta\alpha\nu\rho} =0.\label{f2}%
\end{align}
\begin{remark}
Before proceeding we want to emphasize that the \emph{Eq.(\ref{bianchi1})} and
\emph{Eq.(\ref{bianchi2})} \emph{(}the Bianchi identities\emph{)} do not imply
in general in the validity of the analogous equations for the duals of the
torsion and curvature $2$-forms, i.e., in general\footnote{In particular, a
correct expression for $\mathbf{D}\underset{%
\sslg
}{\star}\mathcal{T}^{\alpha}$ has been found in \cite{rodALB}.}%
\begin{align}
\mathbf{D}\underset{%
\sslg
}{\star}\mathcal{T}^{\alpha} &  \neq\underset{%
\sslg
}{\star}\mathcal{R}_{\beta}^{\alpha}\wedge\theta^{\beta},\label{biii}\\
\mathbf{D}\underset{%
\sslg
}{\star}\mathcal{R}_{\beta}^{\alpha} &  \neq 0.\label{biv}%
\end{align}
\end{remark}

\section{Are $\underset{%
\sslg
}{\star}\mathcal{T}^{\alpha}$ and $\underset{%
\sslg
}{\star}\mathcal{R}_{\beta}^{\alpha}$ the Torsion and Curvature $2$-Forms of
any Connection?}

Despite the fact aforementioned in the last Remark, we may pose the question:
can $\underset{%
\sslg
}{\star}\mathcal{T}^{\alpha}$ and $\underset{%
\sslg
}{\star}\mathcal{R}_{\beta}^{\alpha}$ be the torsion and curvature $2$-forms
of a $%
\slg
$-metric compatible connection, say $D^{\prime}$, which defines on $M$ the
Riemann-Cartan structure $(M,%
\slg
,D^{\prime})$ where also an Einstein like equation is valid? If the answer is
positive, the following set of equations must hold:%
\begin{align}
\underset{%
\sslg
}{\star}\mathcal{G}_{\mu}^{\prime}  &  =\underset{%
\sslg
}{\star}\mathcal{R}_{\mu}^{\prime}-\frac{1}{2}R^{\prime}\underset{%
\sslg
}{\star}\theta_{\mu}=\underset{%
\sslg
}{\star}\mathbf{T}_{\mu}^{\prime}\nonumber\\
\mathbf{D}^{\prime}\mathcal{T}^{\prime\alpha}  &  =\mathcal{R}_{\beta}%
^{\prime\alpha}\wedge\theta^{\beta},\label{f3}\\
\mathbf{D}^{\prime}\mathcal{R}_{\beta}^{\prime\alpha}  &  =0,\nonumber
\end{align}
and since by hypothesis $\underset{%
\sslg
}{\star}\mathcal{T}^{\alpha}=\mathcal{T}^{\prime\alpha}$ and $\underset{%
\sslg
}{\star}\mathcal{R}_{\beta}^{\alpha}=\mathcal{R}_{\beta}^{\prime\alpha}$,
calling $\mathcal{G}_{\mu}^{\ast}=\mathcal{G}_{\mu}^{^{\prime}}$,
$\mathcal{R}_{\mu}^{\ast}$ =$\mathcal{R}_{\mu}^{\prime}$, $R^{\ast}=R^{\prime
}=g^{\mu\beta}R_{\mu\;\rho\beta}^{\ast^{\rho}}$ it must be%
\begin{align}
\underset{%
\sslg
}{\star}\mathcal{G}_{\mu}^{\ast}  &  =\underset{%
\sslg
}{\star}\mathcal{R}_{\mu}^{\ast}-\frac{1}{2}R^{\ast}\underset{%
\sslg
}{\star}\theta_{\mu}=\underset{%
\sslg
}{\star}\mathbf{T}_{\mu}^{\prime}\nonumber\\
\mathbf{D}^{\prime}\underset{%
\sslg
}{\star}\mathcal{T}^{\alpha}  &  =\underset{%
\sslg
}{\star}\mathcal{R}_{\beta}^{\alpha}\wedge\theta^{\beta},\label{f4}\\
\mathbf{D}^{\prime}\underset{%
\sslg
}{\star}\mathcal{R}_{\beta}^{\alpha}  &  =0.\nonumber
\end{align}
or in component form (and obvious notation)
\begin{align}
\left(  R_{\mu\rho\lambda\sigma}+R_{\mu\lambda\sigma\rho}+R_{\mu\sigma
\rho\lambda}+\frac{1}{2}R^{\ast}\epsilon_{\mu\rho\lambda\sigma}\right)   &
=\epsilon_{\rho\lambda\sigma\kappa}T_{\mu}^{\prime\;\kappa}\Longleftrightarrow
G_{\mu\nu}^{\prime}=T_{\mu\nu}^{\prime},\label{f5}\\%
{\displaystyle\sum\limits_{(\mu\alpha\beta)}}
R_{\mu\rho\alpha\beta}^{\ast}  &  =%
{\displaystyle\sum\limits_{(\mu\alpha\beta)}}
\left(  D_{\mu}^{\prime}T_{\rho\alpha\beta}^{\ast}-T_{\mu\beta}^{^{\ast}%
\kappa}T_{\rho\kappa\alpha}^{\ast}\right)  ,\label{f6}\\%
{\displaystyle\sum\limits_{(\mu\nu\rho)}}
D_{\mu}^{\prime}R_{\beta\alpha\nu\rho}^{\ast}  &  =%
{\displaystyle\sum\limits_{(\mu\nu\rho)}}
T_{\nu\mu}^{\ast\kappa}R_{\beta\alpha\kappa\rho}^{\ast} \label{f7}%
\end{align}
Consequently, among the possible constraints in order to have a positive
answer concerning the question in the head of the Section, the following two
non trivial constraints are derived:

\textbf{(a)} Using Eq.(\ref{f10}) and Eq.(\ref{f6}), it follows that%
\begin{equation}
\epsilon_{\rho\lambda\sigma\kappa}T_{\mu}^{\kappa}-\frac{1}{2}R\epsilon
_{\mu\rho\lambda\sigma}=%
{\displaystyle\sum\limits_{(\mu\alpha\beta)}}
\left( D_{\mu}^{\prime}T_{\rho\alpha\beta}^{\ast}-T_{\mu\beta}^{^{\ast}\kappa
}T_{\rho\kappa\alpha}^{\ast}\right) . \label{f8}%
\end{equation}

\textbf{(b) }Using Eq.(\ref{f5}) and Eq.(\ref{f1}) we must have
\begin{equation}
\epsilon_{\rho\lambda\sigma\kappa}T_{\mu}^{\prime\kappa}-\frac{1}{2}R^{\ast
}\epsilon_{\mu\rho\lambda\sigma}=%
{\displaystyle\sum\limits_{(\mu\alpha\beta)}}
\left( D_{\mu}T_{\rho\alpha\beta}-T_{\mu\beta}^{\kappa}T_{\rho\kappa\alpha
}\right) . \label{f9}%
\end{equation}
Let us analyze what those constraints imply if we start with $(M,%
\slg
,\mathring{D}),$ a Lorentzian structure (part of a Lorentzian spacetime
structure) representing a gravitational field in \textit{GRT}. In this case
the second member of Eq.(\ref{f9}) must equal zero, and taking into account
that $\mathring{R}=\mathring{T}^{\prime}:=\mathring{T}_{\kappa}^{\kappa}$ and
$\mathring{R}^{\ast}=-\mathring{T}^{\prime}:=\mathring{T}_{\kappa}%
^{\prime\kappa}$ we get that the structure $(M,%
\slg
,\mathring{D}^{\prime})$ must also be \textit{torsion free} and the following
constraints must hold:
\begin{align}
\epsilon_{\rho\lambda\sigma\kappa}\mathring{T}_{\mu}^{\kappa}  &  =-\frac
{1}{2}\mathring{T}\epsilon_{\mu\rho\lambda\sigma},\qquad\qquad\quad
\epsilon_{\rho\lambda\sigma\kappa}\mathring{T}_{\mu}^{\prime\kappa} =-\frac
{1}{2}\mathring{T}^{\prime}\epsilon_{\mu\rho\lambda\sigma},\label{f20}\\%
{\displaystyle\sum\limits_{(\mu\nu\rho)}}
\mathring{D}_{\mu}^{\prime}R_{\beta\alpha\nu\rho}^{\ast}  &  =%
{\displaystyle\sum\limits_{(\mu\nu\rho)}}
\mathring{D}_{\mu}R_{\beta\alpha\nu\rho}.\nonumber
\end{align}

\section{A Particular Case}

Suppose we have as postulated\footnote{It is obvious from our previous
considerations that the equation $\left(  R_{\mu\rho\lambda\sigma}^{\star
}+R_{\mu\lambda\sigma\rho}^{\star}+R_{\mu\sigma\rho\lambda}^{\star}\right)
=\epsilon_{\rho\lambda\sigma\kappa}T_{\mu}^{\;\kappa}$ presented in
\cite{argueh} as an identity is in general \textit{wrong} and invalidates most
of the conclusions of that paper. Take also notice that in \cite{argueh} it is
defined a Hodge dual with respect to the first pair of indices. However, since
they start from a Lorentzian structure (where torsion is null) we have the
validity of Eq.(\ref{585c}) and so it does not matter in deriving
Eq.(\ref{f10}) taking the dual with respect to the first or second pair of
indices.} in \cite{argueh} a Riemann-Cartan structure where Eq.(\ref{f10}),
Eq.(\ref{f1}), and Eq.(\ref{f11}) read:%
\begin{align}
\left(  R_{\mu\rho\lambda\sigma}^{\star}+R_{\mu\lambda\sigma\rho}^{\star
}+R_{\mu\sigma\rho\lambda}^{\star}\right)   &  =\epsilon_{\rho\lambda
\sigma\kappa}T_{\mu}^{\;\kappa}\Longleftrightarrow G_{\mu\nu}=T_{\mu\nu
},\label{f10'}\\%
{\displaystyle\sum\limits_{(\mu\alpha\beta)}}
R_{\mu\rho\alpha\beta} &  =\epsilon_{\rho\alpha\beta\kappa}\Theta_{\mu
}^{\kappa},\label{f1'}\\%
{\displaystyle\sum\limits_{(\mu\nu\rho)}}
D_{\mu}R_{\beta\alpha\nu\rho} &  =0.\label{f11'}%
\end{align}
It is obvious that we must then have:%
\begin{equation}
R=0,\qquad\qquad\epsilon_{\rho\alpha\beta\kappa}\Theta_{\mu}^{\kappa}=%
{\displaystyle\sum\limits_{(\mu\alpha\beta)}}
\left(  D_{\mu}T_{\rho\alpha\beta}-T_{\mu\beta}^{\kappa}T_{\rho\kappa\alpha
}\right)  ,\qquad\qquad\quad%
{\displaystyle\sum\limits_{(\mu\nu\rho)}}
T_{\nu\mu}^{\kappa}R_{\beta\alpha\kappa\rho}=0,\label{f31}
\end{equation}
\newline and comparing Eq.(\ref{f9}) with Eq.(\ref{f31}) we get%
\begin{equation}
\epsilon_{\rho\lambda\sigma\kappa}T_{\mu}^{\prime\kappa}+\frac{1}{2}T_{\kappa
}^{\prime\kappa}\epsilon_{\mu\rho\lambda\sigma}=\epsilon_{\rho\alpha
\beta\kappa}\Theta_{\mu}^{\kappa}.\label{f33}%
\end{equation}
So a Riemann-Cartan structure satisfying Eq.(\ref{f10'}), Eq.(\ref{f1'}) and
Eq.(\ref{f11'}) is possible only for matter distributions with $T=T_{\kappa
}^{\kappa}=0$ and which obey very stringent constraints.

Also, \cite{argueh} choose as \textquotedblleft dual
equations\textquotedblright\ the following set:%
\begin{align}
\left(  R_{\mu\rho\lambda\sigma}+R_{\mu\lambda\sigma\rho}+R_{\mu\sigma
\rho\lambda}+\frac{1}{2}R^{\ast}\epsilon_{\mu\rho\lambda\sigma}\right)   &
=\epsilon_{\rho\lambda\sigma\kappa}\Theta_{\mu}^{\kappa}\Longleftrightarrow
G_{\mu\nu}^{\ast}=\Theta_{\mu\nu},\\%
{\displaystyle\sum\limits_{(\mu\alpha\beta)}}
R_{\mu\rho\alpha\beta}^{\ast}  &  =\epsilon_{\rho\alpha\beta\kappa}T_{\mu
}^{\kappa},\\%
{\displaystyle\sum\limits_{(\mu\nu\rho)}}
D_{\mu}^{\prime}R_{\beta\alpha\nu\rho}^{\ast}  &  =0
\end{align}
which, of course must imply
\begin{align}
R^{\ast}  &  =0\text{, }\label{l1}\\
\epsilon_{\rho\alpha\beta\kappa}T_{\mu}^{\kappa}  &  =%
{\displaystyle\sum\limits_{(\mu\alpha\beta)}}
\left(  D_{\mu}^{\prime}T_{\rho\alpha\beta}^{\ast}-T_{\mu\beta}^{^{\ast}%
\kappa}T_{\rho\kappa\alpha}^{\ast}\right)  ,\label{l2}\\%
{\displaystyle\sum\limits_{(\mu\nu\rho)}}
T_{\nu\mu}^{\ast\kappa}R_{\beta\alpha\kappa\rho}^{\ast}  &  =0. \label{l3}%
\end{align}
Comparing Eq.(\ref{l3}) to Eq.(\ref{f8}) implies again that $R=0.$ So we end
with the following constraints, necessary for the validity of the equations
proposed in \cite{argueh}:%
\begin{align}
T_{\mu}^{\prime\kappa}  &  =\Theta_{\mu}^{\kappa},\qquad\qquad\qquad T =
T_{\kappa}^{\kappa}=0,\qquad\quad\Theta=\Theta_{\kappa}^{\kappa}=0,\nonumber\\
\epsilon_{\rho\alpha\beta\kappa}\Theta_{\mu}^{\kappa}  &  =%
{\displaystyle\sum\limits_{(\mu\alpha\beta)}}
\left( D_{\mu}T_{\rho\alpha\beta}-T_{\mu\beta}^{\kappa}T_{\rho\kappa\alpha
}\right) ,\qquad\text{ }\epsilon_{\rho\alpha\beta\kappa}T_{\mu}^{\kappa}=%
{\displaystyle\sum\limits_{(\mu\alpha\beta)}}
\left( D_{\mu}^{\prime}T_{\rho\alpha\beta}^{\ast}-T_{\mu\beta}^{^{\ast}\kappa
}T_{\rho\kappa\alpha}^{\ast}\right)  ,\label{ff}\\%
{\displaystyle\sum\limits_{(\mu\nu\rho)}}
T_{\nu\mu}^{\kappa}R_{\beta\alpha\kappa\rho}  &  =0,%
{\displaystyle\sum\limits_{(\mu\nu\rho)}}
T_{\nu\mu}^{\ast\kappa}R_{\beta\alpha\kappa\rho}^{\ast}=0.\nonumber
\end{align}
Such constraints are clearly violated by the examples in \cite{argueh}.

\section{Is there a metric $%
\slg
^{\prime}$ and a metric connection $D^{\prime}$ such that $\underset{%
\sslg
^{\prime}}{\star}\mathcal{T}^{\alpha}$ and $\underset{%
\sslg
^{\prime}}{\star}\mathcal{R}_{\beta}^{\alpha}$ are their Torsion and Curvature
Forms?}

Now, we can also put the question: in which conditions may we conceive that
$\underset{%
\sslg
}{\star}\mathcal{T}^{\alpha}$ and $\underset{%
\sslg
^{\prime}}{\star}\mathcal{R}_{\beta}^{\alpha}$ are the torsion and curvature
$2$-forms of a $%
\slg
^{\prime}$-metric compatible connection, say $D^{\prime}$ which defines on $M$
the Riemann-Cartan structure $(M,%
\slg
^{\prime},D^{\prime})$ where an Einstein like equation holds, i.e., (with
obvious notation) the validity of the following set of equations ($R^{\prime
}=g^{\prime\mu\beta}R_{\mu\rho\beta}^{\prime^{\rho}}$):
\begin{align}
\underset{%
\sslg
^{\prime}}{\star}\mathcal{G}_{\mu}^{\prime}  &  =\underset{%
\sslg
^{\prime}}{\star}\mathcal{R}_{\mu}^{\prime}-\frac{1}{2}R^{\prime}\underset{%
\sslg
^{\prime}}{\star}\theta_{\mu}=\underset{%
\sslg
^{\prime}}{\star}\mathbf{T}_{\mu}^{\prime}\nonumber\\
\mathbf{D}^{\prime}\mathcal{T}^{\prime\alpha}  &  =\mathcal{R}_{\beta}%
^{\prime\alpha}\wedge\theta^{\beta},\label{f21}\\
\mathbf{D}^{\prime}\mathcal{R}_{\beta}^{\prime\alpha}  &  =0.\nonumber
\end{align}
Since by hypothesis we must have $\underset{%
\sslg
}{\star}\mathcal{T}^{\alpha}=\mathcal{T}^{\prime\alpha}$ and $\underset{%
\sslg
}{\star}\mathcal{R}_{\beta}^{\alpha}=\mathcal{R}_{\beta}^{\prime\alpha},$
calling $R^{\ast}=g^{\prime\mu\beta}R_{\mu\rho\beta}^{\ast^{\rho}}$ \ the set
of Eqs.(\ref{f21}) must be equal to:%
\begin{align}
\underset{%
\sslg
^{\prime}}{\star}\mathcal{G}_{\mu}^{\ast}  &  =\underset{%
\sslg
^{\prime}}{\star}\mathcal{R}_{\mu}^{\ast}-\frac{1}{2}R^{\ast}\underset{%
\sslg
^{\prime}}{\star}\theta_{\mu}=\underset{%
\sslg
^{\prime}}{\star}\mathbf{T}_{\mu}^{\prime}\nonumber\\
\mathbf{D}^{\prime}\underset{%
\sslg
}{\star}\mathcal{T}^{\prime\alpha}  &  =\underset{%
\sslg
}{\star}\mathcal{R}_{\beta}^{\alpha}\wedge\theta^{\beta},\label{f22}\\
\mathbf{D}^{\prime}\underset{%
\sslg
}{\star}\mathcal{R}_{\beta}^{\prime\alpha}  &  =0.\nonumber
\end{align}
which are similar but not identical to the set given by Eq.(\ref{f4}). Due to
their complexity we shall not inspect the nature of those equations solutions,
a problem postponed for another publication.

\begin{remark}
The constraints concerned in this case are more involved than in the previous
case, but we want to emphasize here that if we start with $(M,%
\slg
,\mathring{D}),$ a Lorentzian structure \emph{(}part of a Lorentzian spacetime
structure\emph{) }representing a gravitational field in \textit{GRT}, the
structure $(M,%
\slg
^{\prime},\mathring{D}^{\prime})$ will be also torsion free. Here we recalled
that \cite{dowro} investigated long ago a similar problem (but only in the
linear approximation) and found a positive answer for the question at the head
of this Section.
\end{remark}

\section{Conclusions}

In this paper we present the \textit{correct} constraints that must be
satisfied by any theory (in a $4$-dimensional manifold) that intends to
provide a dual presentation of the gravitational field equations for a general
Riemann-Cartan theory. We compare our results with some ones proposed by
authors quoted in the introduction and present some constructive criticisms.
We hope that since the subject of duality becomes each day more important in
non, e.g., abelian gauge theories, gravity and $M$-theory that our results
shall become appreciated. \appendix

\section{Exterior Covariant Derivative $\mathbf{D}$}

Sometimes Eqs.(\ref{559}) are written by some authors as:
\begin{align}
\mathbf{D}\theta^{\rho}  &  =\mathcal{T}^{\rho},\qquad\qquad\quad
\text{\textquotedblleft\ }\mathbf{D}\omega_{\mu}^{\rho} =\mathcal{R}_{\mu
}^{\rho}.\text{\textquotedblright} \label{559b}%
\end{align}
and $\mathbf{D}: \sec\Lambda T^{\ast}M\rightarrow\sec\Lambda T^{\ast}M$ is
said to be the \textit{exterior covariant derivative} related to the
connection $D$. The second of Eqs.(\ref{559b}) has been printed with quotation
marks due to the fact that it is not a correct equation. Indeed, a
\textit{legitimate} exterior covariant derivative operator\footnote{Sometimes
also called exterior covariant differential.} is a concept that can be defined
for $(p+q)$-indexed $r$-form fields\footnote{Which is not the case of the
connection $1$-forms $\omega_{\beta}^{\alpha}$, despite the name. More
precisely, the $\omega_{\beta}^{\alpha}$ are not true indexed forms, i.e.,
there does not exist a tensor field $\mathbf{\omega}$ such that
$\mathbf{\omega(}e_{i},e_{\beta},\vartheta^{\alpha})=$ $\omega_{\beta}%
^{\alpha}(e_{i}).$} as follows. Suppose that $X\in\sec T_{p}^{r+q}M$ and let $
X_{\nu_{1}\ldots\nu_{q}}^{\mu_{1}\ldots\mu_{p}}\in\sec\Lambda^{r}T^{\ast}M,$
such that for $v_{i}\in\sec TM,$ {$i=0,1,2,\ldots,r$}, then $ X_{\nu_{1}%
\ldots\nu_{q}}^{\mu_{1}\ldots\mu_{p}}(v_{1},\ldots,v_{r})=X(v_{1},\ldots
,v_{r},e_{\nu_{1}},\ldots,e_{\nu_{q}},\theta^{\mu_{1}},\ldots,\theta^{\mu_{p}%
})$. The exterior covariant differential $\mathbf{D}$ of $X_{\nu_{1}\ldots
\nu_{q}}^{\mu_{1}\ldots\mu_{p}}$ on a manifold with a general connection $D$
is the mapping%
\begin{equation}
\mathbf{D:}\sec\Lambda^{r}T^{\ast}M\rightarrow\sec\Lambda^{r+1}T^{\ast
}M\text{, }\qquad0\leq r\leq4, \label{559new2}%
\end{equation}
such that\footnote{As usual the inverted hat over a symbol (in
Eq.(\ref{559new3})) means that the corresponding symbol is missing in the
expression.}%
\begin{align}
&  (r+1)\mathbf{D}X_{\nu_{1}\ldots\nu_{q}}^{\mu_{1}\ldots\mu_{p}}(v_{0}%
,v_{1},\ldots,v_{r})\nonumber\\
&  =\sum\limits_{\nu=0}^{r}(-1)^{\nu}D_{\mathbf{e}_{\nu}}X(v_{0},v_{1}%
,\ldots,\check{v}_{\nu},\ldots,v_{r},%
\sle
_{\nu_{1}},\ldots,%
\sle
_{\nu_{q}},\theta^{\mu_{1}},\ldots,\theta^{\mu_{p}})\nonumber\\
&  \quad-\sum\limits_{0\leq\lambda,\varsigma\,\leq r}(-1)^{\nu+\varsigma
}X(\mathbf{T(}v_{\lambda},v_{\varsigma}),v_{0},v_{1},\ldots,\check{v}%
_{\lambda},\ldots,\check{v}_{\varsigma},\ldots,v_{r},%
\sle
_{\nu_{1}},\ldots,%
\sle
_{\nu_{q}},\theta^{\mu_{1}},\ldots,\theta^{\mu_{p}}). \label{559new3}%
\end{align}
Then, we may verify that
\begin{align}
\hspace{-0.3cm}\mathbf{D}X_{\nu_{1}\ldots\nu_{q}}^{\mu_{1}\ldots\mu_{p}}
=dX_{\nu_{1}\ldots\nu_{q}}^{\mu_{1}\ldots\mu_{p}}+\omega_{\mu_{s}}^{\mu_{1}%
}\wedge X_{\nu_{1}\ldots\nu_{q}}^{\mu_{s}\ldots\mu_{p}}+\cdots+\text{ }%
\omega_{\mu_{s}}^{\mu_{1}}\wedge X_{\nu_{1}\ldots\nu_{q}}^{\mu_{1}\ldots
\mu_{p}} -\omega_{\nu_{1}}^{\nu_{s}}\wedge X_{\nu_{s}\ldots\nu_{q}}^{\mu
_{1}\ldots\mu_{p}}-\cdots-\text{ }\omega_{\mu_{s}}^{\mu_{1}}\wedge X_{\nu
_{1}\ldots\nu_{s}}^{\mu_{1}\ldots\mu_{p}}.\label{559new4}%
\end{align}

\begin{remark}
Note that if \emph{Eq.(\ref{559new4})} is applied on any one of the connection
$1$-forms $\omega_{\nu}^{\mu}$ we would get $\mathbf{D}\omega_{\nu}^{\mu
}=d\omega_{\nu}^{\mu}+\omega_{\alpha}^{\mu}\wedge\omega_{\nu}^{\alpha}%
-\omega_{\nu}^{\alpha}\wedge\omega_{\alpha}^{\mu}$. So, we see that the symbol
$\mathbf{D}\omega_{\nu}^{\mu}$ in \emph{Eq.(\ref{559b})}, supposedly defining
the curvature $2$-forms, is simply wrong, despite this being an equation
printed in many Physics textbooks and many professional articles.
\end{remark}

\subsection{ Properties of $\mathbf{D}$}

The exterior covariant derivative $\mathbf{D}$ satisfy the following properties:

\textbf{(a)} For any $X^{J}$ $\in\sec\Lambda^{r}T^{\ast}M$ and $Y^{K}$
$\in\sec\Lambda^{s}T^{\ast}M$ are sets of indexed forms\footnote{Multi indices
are here represented by $J$ and $K$.}, then%
\begin{equation}
\mathbf{D(}X^{J}\wedge Y^{K}\mathbf{)=D}X^{J}\wedge Y^{K}+(-1)^{rs}X^{J}%
\wedge\mathbf{D}Y^{K}. \label{559new5}%
\end{equation}

\textbf{(b)} For any $X^{\mu_{1}\ldots\mu_{p}}\in\sec\Lambda^{r}T^{\ast}M$
then%
\begin{equation}
\mathbf{DD}X^{\mu_{1}\ldots\mu_{p}}=dX^{\mu_{1}\ldots\mu_{p}}+\mathcal{R}%
_{\mu_{s}}^{\mu_{1}}\wedge X^{\mu_{s}\ldots\mu_{p}}+\cdots+\mathcal{R}%
_{\mu_{s}}^{\mu_{p}}\wedge X^{\mu_{1}\ldots\mu_{s}}. \label{559new6}%
\end{equation}

\textbf{(c)} For any metric-compatible connection $D$ if $g=g_{\mu\nu}%
\theta^{\mu}\otimes\theta^{\nu}$ then, $\mathbf{D}g_{\mu\nu}=0$.

\section{Relation Between the Riemann Curvature Tensors of the Levi-Civita
Connection of $\mathring{%
\slg
}$ and a $%
\slg
$-compatible Riemann-Cartan Connection}

Let $(M,\mathring{%
\slg
},\mathring{D}\mathbf{)}$ and $(M,%
\slg
,D)$ be respectively a Lorentzian and a Riemann-Cartan structure\footnote{Note
that $(M,%
\slg
,{D}\mathbf{)}$ and $(M,\mathring{%
\slg
},\mathring{D})$ are in general \textit{Riemann-Cartan-Weyl} structures. More
general formulas relating two arbitrary general connections may be found,
e.g., in \cite{rodcap2007}.} on the same manifold $M$ such that%
\begin{equation}
\mathring{D}\mathring{%
\slg
}=0,\qquad D%
\slg
=0,
\end{equation}
with the \textit{nonmetricity} of $D$ associated with $\mathring{%
\slg
}$ being given by $ \mathbf{Q:=-}D\mathring{%
\slg
}.$ Let moreover the connection coefficients of $\mathring{D}$ and $D$ in the
arbitrary bases dual bases $\{%
\sle
_{\alpha}\}$ and $\{\theta^{\rho}\}$ for $TU\subset TM$ and $T^{\ast}U\subset
T^{\ast}M$ be:%

\begin{equation}
\mathring{D}_{\partial_{\alpha}}\theta^{\rho}=-\mathring{\Gamma}_{\alpha\beta
}^{\rho}\theta^{\beta}, \qquad\quad D_{\partial_{\alpha}}\theta^{\rho
}=-L_{\alpha\beta}^{\rho}\theta^{\beta},
\end{equation}
and $ Q_{\alpha\beta\sigma}=-D_{\alpha}\mathring{g}_{\beta\sigma}.$ Define the
components of the \textit{strain tensor} of the connection $D$ (associated
with $\mathring{D})$ by
\begin{equation}
S_{\alpha\beta}^{\rho}=(L_{\alpha\beta}^{\rho}+L_{\alpha\beta}^{\rho
})-(\mathring{\Gamma}_{\alpha\beta}^{\rho}+\mathring{\Gamma}_{\alpha\beta
}^{\rho})
\end{equation}
It is trivially established that
\begin{equation}
L_{\alpha\beta}^{\rho}=\mathring{\Gamma}_{\alpha\beta}^{\rho}+\frac{1}%
{2}T_{\alpha\beta}^{\rho}+\frac{1}{2}S_{\alpha\beta}^{\rho}. \label{1041}%
\end{equation}
where $\mathring{\Gamma}_{\alpha\beta}^{\rho}$ are the components of the
Levi-Civita connection of $%
\slg
$ and $T_{\alpha\beta}^{\rho}$ are the components of the torsion tensor of
$D$\footnote{More details may be found, e.g., in \cite{rodcap2007}.}.

Eq.(\ref{1041}) can be used to relate the covariant derivatives with respect
to the connections $\mathring{D}$ and $D$ of any tensor field on the manifold.
In particular, recalling that $\mathring{D}_{\!\alpha}\mathring{g}%
_{\beta\sigma}=$\texttt{ }$%
\sle
_{\alpha}(\mathring{g}_{\beta\sigma})-\mathring{g}_{\mu\sigma}\mathring
{\Gamma}_{\alpha\beta}^{\mu}-\mathring{g}_{\beta\mu}\mathring{\Gamma}%
_{\alpha\sigma}^{\mu}=0$, we get the expression of the nonmetricity tensor of
$D$ in terms of the torsion and the strain, namely,
\begin{equation}
Q_{\alpha\beta\sigma}=\frac{1}{2}(\mathring{g}_{\mu\sigma}T_{\alpha\beta}%
^{\mu}+\mathring{g}_{\beta\mu}T_{\alpha\sigma}^{\mu})+\frac{1}{2}(\mathring
{g}_{\mu\sigma}S_{\alpha\beta}^{\mu}+\mathring{g}_{\beta\mu}S_{\alpha\sigma
}^{\mu}). \label{1077}%
\end{equation}
Eq.(\ref{1077}) can be inverted to yield the expression of the strain in terms
of the torsion and the nonmetricity. We get:
\begin{equation}
S_{\alpha\beta}^{\rho}=\mathring{g}^{\rho\sigma}(Q_{\alpha\beta\sigma
}+Q_{\beta\sigma\alpha}-Q_{\sigma\alpha\beta})-\mathring{g}^{\rho\sigma
}(\mathring{g}_{\beta\mu}T_{\alpha\sigma}^{\mu}+\mathring{g}_{\alpha\mu
}T_{\beta\sigma}^{\mu}). \label{1084}%
\end{equation}
From Eq.(\ref{1077}) and~Eq.(\ref{1084}) it is clear that nonmetricity and
strain can be used interchangeably in the description of the geometry of a
Riemann-Cartan-Weyl space. In particular, we have the relation:
\begin{equation}
Q_{\alpha\beta\sigma}+Q_{\sigma\alpha\beta}+Q_{\beta\sigma\alpha}%
=S_{\alpha\beta\sigma}+S_{\sigma\alpha\beta}+S_{\beta\sigma\alpha}%
,\qquad\text{where $S_{\alpha\beta\sigma}=\mathring{g}_{\rho\sigma}%
S_{\alpha\beta}^{\rho}$.}%
\end{equation}

In order to simplify our next equations, let us introduce the notation:
\begin{equation}
K_{\alpha\beta}^{\!\rho}=L_{\alpha\beta}^{\rho}-\mathring{\Gamma}_{\alpha
\beta}^{\rho}=\frac{1}{2}(T_{\alpha\beta}^{\rho}+S_{\alpha\beta}^{\rho}).
\label{1168}%
\end{equation}
From Eq.(\ref{1084}) it follows that:
\begin{align}
K_{\alpha\beta}^{\!\rho}  &  =-\frac{1}{2}\mathring{g}^{\rho\sigma
}(D_{\!\alpha}\mathring{g}_{\beta\sigma}+D_{\!\beta}\mathring{g}_{\sigma
\alpha}-D_{\!\sigma}\mathring{g}_{\alpha\beta}) -\frac{1}{2}\mathring{g}%
^{\rho\sigma}(\mathring{g}_{\mu\alpha}T_{\sigma\beta}^{\mu}+\mathring{g}%
_{\mu\beta}T_{\sigma\alpha}^{\mu}-\mathring{g}_{\mu\sigma}T_{\alpha\beta}%
^{\mu}). \label{1130}%
\end{align}
Note also that for $D\mathring{%
\slg
}=0$, $K_{\alpha\beta}^{\rho}$ is the so-called \textit{contorsion tensor}.

Returning to Eq.(\ref{1041}), we obtain now the relation between the curvature
tensor $R_{\mu}{}^{\rho}{}_{\!\alpha\beta}$ associated with the connection $D$
and the Riemann curvature tensor $\mathring{R}{}_{\mu}{}^{\rho}{}%
_{\!\alpha\beta}$ of the Levi-Civita connection $D$ associated with the metric
$%
\slg
$. We get, by a straightforward calculation:
\begin{equation}
R_{\mu}{}^{\rho}{}_{\!\alpha\beta}=\mathring{R}_{\mu}{}^{\rho}{}%
_{\!\alpha\beta}+J_{\mu}{}^{\rho}{}_{\![\alpha\beta]}, \label{1070}%
\end{equation}
where:
\begin{equation}
J_{\!\mu}{}^{\rho}{}_{\!\alpha\beta}=\mathring{D}_{\!\alpha}K_{\beta\mu
}^{\!\rho}-K_{\beta\sigma}^{\!\rho}K_{\alpha\mu}^{\!\sigma}=D_{\!\alpha
}K_{\beta\mu}^{\!\rho}-K_{\alpha\sigma}^{\!\rho}K_{\beta\mu}^{\!\sigma
}+K_{\alpha\beta}^{\!\sigma}K_{\sigma\mu}^{\!\rho}. \label{1070a}%
\end{equation}
Multiplying both sides of Eq.(\ref{1070}) by $\frac{1}{2}\theta^{\alpha}%
\wedge\theta^{\beta}$ we get:
\begin{equation}
\mathcal{R}_{\mu}^{\rho}=\mathcal{\mathring{R}}_{\mu}^{\rho}+\mathfrak{J}%
_{\mu}^{\rho},\qquad\text{ where $\mathfrak{J}_{\mu}^{\rho}=\frac{1}%
{2}J_{\!\mu}{}^{\rho}{}_{\![\alpha\beta]}\theta^{\alpha}\wedge\theta^{\beta}%
.$}%
\end{equation}
From Eq.(\ref{1070}) we also get the relation between the Ricci tensors of the
connections $D$ and $\mathring{D}$. The \textit{Ricci tensor} is defined by
\begin{align}
Ricci  &  =R_{\mu\alpha}dx^{\mu}\otimes dx^{\nu},\qquad\text{where}\qquad\quad
R_{\mu\alpha}:=R_{\mu}{}^{\rho}{}_{\!\alpha\rho} \label{ricci}%
\end{align}
Then, we have
\begin{equation}
R_{\mu\alpha}=\mathring{R}_{\mu\alpha}+J_{\mu\alpha}, \label{1174}%
\end{equation}
with
\begin{align}
J_{\mu\alpha}  &  =\mathring{D}_{\alpha}K_{\rho\mu}^{\!\rho}-\mathring
{D}_{\rho}K_{\alpha\mu}^{\!\rho}+K_{\alpha\sigma}^{\!\rho}K_{\rho\mu
}^{\!\sigma}-K_{\rho\sigma}^{\!\rho}K_{\alpha\mu}^{\!\sigma}\nonumber\\
&  =D_{\alpha}K_{\rho\mu}^{\!\rho}-D_{\rho}K_{\alpha\mu}^{\!\rho}%
-K_{\sigma\alpha}^{\!\rho}K_{\rho\mu}^{\!\sigma}+K_{\rho\sigma}^{\!\rho
}K_{\alpha\mu}^{\!\sigma}.
\end{align}
Observe that since the connection $D$ is arbitrary, its Ricci tensor will be
\textit{not} be symmetric in general. Then, since the Ricci tensor
$\mathring{R}_{\mu\alpha}$ of $\mathring{D}$ is necessarily symmetric, we can
split Eq.(\ref{1174}) into:
\begin{equation}%
\begin{tabular}
[c]{c}%
$R_{[\mu\alpha]}=J_{[\mu\alpha]},$ \qquad\quad$R_{(\mu\alpha)}=\mathring
{R}_{(\mu\alpha)}+J_{(\mu\alpha)}.$%
\end{tabular}
\ \ \ \ \ \ \ \label{1190}%
\end{equation}

\section{Some Important Identities}

Let $(M,%
\slg
)$ be a manifold and a Lorentzian metric as defined in Section 1. Let
moreover  $\Lambda^{p}T^{\ast}M$ ($p=0,1,2,3,4$) be the bundle of homogeneous
$p$-form fields and $\Lambda T^{\ast}M=\oplus_{p=0}^{4}\Lambda^{p}T^{\ast}M$
the bundle of non homogeneous forms fields.  We define in $T^{\ast}M$ a metric
field \texttt{g} $\in\sec T_{0}^{2}M$ such that concerning the general bases
$\{%
\sle
_{\mu}\}$ and $\{\theta^{\mu}\}$ introduced in Section 1, if $%
\slg
=g_{\mu\nu}\theta^{\mu}\otimes\theta^{\nu}$ and \texttt{g }$=g^{\mu\nu}%
\sle
_{\mu}\otimes%
\sle
_{\nu}$ then $g^{\mu\alpha}g_{\alpha\nu}=\delta_{\nu}^{\mu}$. In $\Lambda
T^{\ast}M$ we introduce a scalar product
\begin{equation}
\cdot:\Lambda T^{\ast}M\times\Lambda T^{\ast}M\rightarrow\Lambda T^{\ast}M
\label{sp1}%
\end{equation}
such that if $A,B\in\sec\Lambda^{r}T^{\ast}M$ are simple homogeneous $r$-forms
with $A=u_{1}\wedge\cdots\wedge u_{r}$ and $B=v_{1}\wedge\cdots\wedge v_{r}$,
$u_{i},v_{j}\in\sec\Lambda^{1}T^{\ast}M$ then $A\cdot B=\det(\mathtt{g}%
\text{\texttt{(}}u_{i},v_{j}\text{\texttt{)}})$, where $(\mathtt{g}$%
\texttt{(}$u_{i},v_{j}$\texttt{)}$)$ means the matrix with entries
$(\mathtt{g}$\texttt{(}$u_{i},v_{j}$\texttt{)}$)$. This scalar product is then
extended by linearity and orthogonality to all $\Lambda T^{\ast}M$,  and
$A\cdot B=0$ if $A\in\sec\Lambda^{r}T^{\ast}M$, and $B\in\sec\Lambda
^{s}T^{\ast}M$ with $r\neq s$. Also, if $a,b\in\sec\Lambda
^{0}T^{\ast}M$ then $a\cdot b=ab$, the product of functions.

If the metric manifold $(M,%
\slg
)$ is also endowed with an \textit{orientation}, i.e., a {volume }$4${-vector}
$\tau_{\mathring{%
\sslg
}}\in\Lambda^{4}T^{\ast}M$ such that $\tau_{\mathring{%
\sslg
}}\cdot\tau_{\mathring{%
\sslg
}}=-1,$ then a natural isomorphism between sections of $\Lambda^{r}T^{\ast}M$
and $\Lambda^{4-r}T^{\ast}M$ ($r=0,\ldots,4$) can be introduced. The
\textit{Hodge star operator} (or \textit{Hodge dual}) is the linear mapping
$\underset{%
\sslg
}{\star}:\sec\Lambda^{r}T^{\ast}M\rightarrow\sec\Lambda^{4-r}T^{\ast}M$ implicitly defined by 
\begin{equation}
A\wedge\underset{%
\sslg
}{\star}B=(A\cdot B\,)\tau_{\mathring{%
\sslg
}},\label{hodge}%
\end{equation}
for every $A,B\in\Lambda^{r}T^{\ast}M$. Of course, this operator is naturally
extended to an isomorphism $\underset{%
\sslg
}{\star}:\sec\Lambda T^{\ast}M\rightarrow\sec\Lambda T^{\ast}M$ by linearity.
The inverse $\underset{%
\sslg
}{\star}^{-1}:\sec\Lambda^{r}T^{\ast}M\rightarrow\sec\Lambda^{4-n}T^{\ast}M$
of the Hodge star operator is given by $\underset{%
\sslg
}{\star}^{-1}=-(-1)^{r(4-r)}\underset{%
\sslg
}{\star}$. For any $A,B\in\sec\Lambda T^{\ast}M$
\begin{equation}
A\cdot B=\langle\tilde{A}\text{ }B\rangle_{0}=\langle A\text{ }\tilde
{B}\rangle_{0}=B\cdot A,\label{T.55}%
\end{equation}
where $\tilde{A}$ means the reverse of $A$. If $A=u_{1}\wedge\cdots\wedge
u_{r}$ then $\tilde{A}=u_{r}\wedge\cdots\wedge u_{1}$ and $\langle
\;\;\rangle_{0}:\sec\Lambda T^{\ast}M\rightarrow\sec\Lambda^{0}T^{\ast}M$ is
the projection of a general non homogeneous form into the $\Lambda^{0}T^{\ast
}M$ part.

\begin{remark}
Suppose that $\{\varepsilon_{i}\}$ is an orthonormal basis of
$\Lambda^{1}T^{\ast}M$ and $\{\varepsilon^{j}\}$ is reciprocal basis, i.e.,
$\varepsilon_{i}\cdot\varepsilon^{k}=\delta_{i}^{k}$. Then, any $Y\in
\sec\Lambda^{p}T^{\ast}M$ can be written as
\begin{align}
Y  &  =\frac{1}{p!}Y^{j_{1\ldots}}{}^{j_{p}}\varepsilon_{j_{1}}\wedge
\cdots\wedge\varepsilon_{j_{p}} =\frac{1}{p!}Y_{j_{1}}..._{j_{p}}%
\varepsilon^{j_{1}}\wedge\cdots\wedge\varepsilon^{j_{p}}. \label{TEXPANSION0}%
\end{align}
and
\begin{equation}
Y^{j_{1\ldots}}{}^{j_{p}}=Y\cdot(\varepsilon^{j_{1}}\wedge\cdots
\wedge\varepsilon^{j_{p}}),\qquad\qquad\quad Y_{j_{1\ldots}}{}_{j_{p}}%
=Y\cdot(\varepsilon_{j_{1}}\wedge\cdots\wedge\varepsilon_{j_{p}}).
\label{TEXPANSION}%
\end{equation}

\end{remark}

We define the right and left contractions of non homogeneous differential
forms as follows. For arbitrary multiforms $X,Y,Z\in\sec$ $\Lambda T^{\ast}M$,
the left $(\underset{%
\sslg
}{\lrcorner})$ and right $(\underset{%
\sslg
}{\llcorner})$ contractions of $X$ and $Y$ are the mappings $\underset{%
\sslg
}{\lrcorner}:\sec\Lambda T^{\ast}M\times\sec\Lambda T^{\ast}M\rightarrow
\sec\Lambda T^{\ast}M$ and $\underset{%
\sslg
}{\llcorner}:\sec\Lambda T^{\ast}M\times\sec\Lambda T^{\ast}M\rightarrow
\sec\Lambda T^{\ast}M$ such that%
\begin{align}
(X\underset{%
\sslg
}{\lrcorner}Y)\cdot Z  &  =Y\cdot(\tilde{X}\wedge Z),\qquad\qquad
\quad(X\underset{%
\sslg
}{\llcorner}Y)\cdot Z = X\cdot(Z\wedge\tilde{Y}). \label{T49}%
\end{align}
These contracted products $\underset{%
\sslg
}{\lrcorner}$ and $\underset{%
\sslg
}{\llcorner}$ are inner derivations on $\Lambda T^{\ast}M$. Sometimes the
contractions are called interior products. Both contract products satisfy the
left and right distributive laws but they are \emph{not }associative.
Now some important properties of the contractions used in the calculations of
the text are presented:

(i) For any $a,b\in\sec\Lambda^{0}T^{\ast}M,$ and $Y\mathbb{\in}\sec\Lambda
T^{\ast}M$%
\begin{align}
a\underset{%
\sslg
}{\lrcorner}b  &  =a\underset{%
\sslg
}{\llcorner}b=ab\text{ (product of functions),}\quad\quad\quad a\underset{%
\sslg
}{\lrcorner}Y =Y\underset{%
\sslg
}{\llcorner}a=aY\text{ (multiplication by scalars).} \label{T50}%
\end{align}

(ii) If $a,b_{1},\ldots,b_{k}\in\sec\Lambda T^{\ast}M$ then $ a\underset{%
\sslg
}{\lrcorner}(b_{1}\wedge\cdots\wedge b_{k})=\sum\limits_{j=1}^{k}%
(-1)^{j+1}(a\cdot b_{j})b_{1}\wedge\cdots\wedge\check{b}_{j}\wedge\cdots\wedge
b_{k}$, where the symbol $\check{b}_{j}$ means that the $b_{j}$ factor does
not appear in the $j$-term of the sum.

(iii) For any $Y_{j}\in\sec\Lambda^{j}T^{\ast}M$ and $Y_{k}\in\sec\Lambda
^{k}T^{\ast}M$ with $j\leq k$%
\begin{equation}
Y_{j}\underset{%
\sslg
}{\lrcorner}Y_{k}=(-1)^{j(k-j)}Y_{k}\underset{%
\sslg
}{\llcorner}Y_{j}. \label{T51}%
\end{equation}

(iv) For any $Y_{j}\in\sec\Lambda^{j}T^{\ast}M$ and $Y_{k}\in\sec\Lambda
^{k}T^{\ast}M$%
\begin{align}
Y_{j}\underset{%
\sslg
}{\lrcorner}Y_{k}  &  =0,\text{ if }j>k, \qquad\qquad Y_{j}\underset{%
\sslg
}{\llcorner}Y_{k} =0,\text{ if }j<k. \label{T52}%
\end{align}

(v) For any $X_{k},Y_{k}\in\sec\Lambda^{k}T^{\ast}M$, then $ X_{k}\underset{%
\sslg
}{\lrcorner}Y_{k}=Y_{k}\underset{%
\sslg
}{\llcorner}Y_{k}=\tilde{X}_{k}\cdot Y_{k}=X_{k}\cdot\tilde{Y}_{k}.$

(vi) For any $v\in\sec\Lambda^{1}T^{\ast}M$ and $X,Y\in\sec\Lambda T^{\ast}M$,
then $v\underset{%
\sslg
}{\lrcorner}(X\wedge Y)=(v\underset{%
\sslg
}{\lrcorner}X)\wedge Y+\hat{X}\wedge(v\underset{%
\sslg
}{\lrcorner}Y).$ Also, if $A,B\in\sec\Lambda^{k}T^{\ast}M$ then$ A\underset{%
\sslg
}{\lrcorner}(B\underset{%
\sslg
}{\lrcorner}C)=(A\wedge B)\underset{%
\sslg
}{\lrcorner}C,$ and $A\underset{%
\sslg
}{\llcorner}(B\underset{%
\sslg
}{\llcorner}C) =A\underset{%
\sslg
}{\llcorner}(B\wedge C)$.

{(vii)} if $A,B\in\sec\Lambda T^{\ast}M$ then%
\begin{align}
(A\underset{%
\sslg
}{\lrcorner}B)\cdot C  &  =B\cdot(\tilde{A}\wedge C),\qquad\qquad(B%
\sslg
{\underset{%
\sslg
}{\llcorner}}A)\cdot C =B\cdot(C\wedge\tilde{A}). \label{tn954}%
\end{align}
Finally we present some important identities involving contractions and the
Hodge dual. Let $A_{r}\in\sec\Lambda^{r}T^{\ast}M$ and $B_{s}\in\sec
\Lambda^{s}T^{\ast}M$, $r,s\geq0$:
\begin{equation}%
\begin{array}
[c]{l}%
A_{r}\wedge\underset{%
\sslg
}{\star}B_{s}=B_{s}\wedge\underset{%
\sslg
}{\star}A_{r}\quad r=s;\qquad\qquad A_{r}\cdot\underset{%
\sslg
}{\star}B_{s}=B_{s}\cdot\underset{%
\sslg
}{\star}A_{r};\quad r+s=n,\\
A_{r}\wedge\underset{%
\sslg
}{\star}B_{s}=(-1)^{r(s-1)}\underset{%
\sslg
}{\star}(\tilde{A}_{r}\underset{%
\sslg
}{\lrcorner}B_{s});\quad r\leq s,\\
A_{r}\underset{%
\sslg
}{\lrcorner}\underset{%
\sslg
}{\star}B_{s}=(-1)^{rs}\underset{%
\sslg
}{\star}(\tilde{A}_{r}\wedge B_{s});\quad r+s\leq n,\\
\underset{%
\sslg
}{\star}A_{r}=\tilde{A}_{r}\underset{%
\sslg
}{\lrcorner}\tau_{\mathring{%
\sslg
}},\qquad\qquad\underset{%
\sslg
}{\star}\tau_{\mathring{%
\sslg
}}=-1,\qquad\qquad\underset{%
\sslg
}{\star}1=\tau_{\mathring{%
\sslg
}}.
\end{array}
\label{440new}%
\end{equation}

\end{document}